\documentclass[11pt]{article}

\usepackage{amsmath}
\usepackage{amsthm}
\usepackage{amssymb}
\usepackage{amsfonts}
\usepackage{subfigure}
\usepackage{color}
\usepackage{makeidx}
\usepackage{hyperref}
\usepackage{dsfont}
\usepackage{mathrsfs}

\newcommand{\ignore}[1]{}

\newtheorem{theorem}{Theorem}

\newtheorem{lemma}{Lemma}

\renewcommand{\Pr}{{\bf Pr}}

\newcommand{\dist}{{\rm dist}}
\newcommand{\FF}{\mathbb{F}}
\newcommand{\labell}[1]{\label{#1}}

\newcommand{\cB}{{\cal B}}

\newcommand{\cM}{{\cal M}}

\newcommand{\cF}{{\cal F}}

\newcommand{\bfi}{{\boldsymbol i}}

\begin{document}

\title{Linear time Constructions of some \\ $d$-Restriction Problems}
\author{ Nader H. Bshouty
\\ Dept. of Computer Science\\ Technion, Haifa, 32000}


%
\maketitle

\begin{abstract}
We give new linear time globally explicit
constructions of some $d$-restriction problems that follows
from the techniques used in \cite{A86,NSS95,PR11}.
\end{abstract}

\section{Introduction}
A $d$-{\it restriction problem} \cite{NSS95,AMS06,B14} is a problem of the
following form:

\noindent {\bf Given} an alphabet $\Sigma$ of size $|\Sigma|=q$, an
integer $n$ and a class $\cM$ of nonzero functions $f:\Sigma^d\to
\{0,1\}$.

\noindent {\bf  Find} a small set $A\subseteq \Sigma^n$ such that: For every $1\le i_1< i_2<\cdots < i_d\le n$
and $f\in \cM$ there is $a\in A$ such that
$f(a_{i_1},\ldots,a_{i_d})\not=0$.

A $(1-\epsilon)$-{\it dense} $d$-{\it restriction problem} is a problem of the
following form:

\noindent {\bf Given} an alphabet $\Sigma$ of size $|\Sigma|=q$, an
integer $n$ and a class $\cM$ of nonzero functions $f:\Sigma^d\to
\{0,1\}$.

\noindent {\bf  Find} a small set $A\subseteq \Sigma^n$ such that: For every $1\le i_1< i_2<\cdots < i_d\le n$
and $f\in \cM$
$$\Pr_{a\in A}[f(a_{i_1},\ldots,a_{i_d})\not=0]> 1-\epsilon$$
where the probability is over the choice
of $a$ from the uniform distribution on $A$.

We give new constructions for the following three
($(1-\epsilon)$-dense) $d$-restriction problems: Perfect hash family, cover-free family and separating hash family.

A
construction is {\it global explicit} if it runs in deterministic polynomial time in the size
of the construction.
A {\it local explicit construction} is
a construction where one can find any bit in the construction
in time poly-log in the size of the construction.
The constructions in this paper are linear time global explicit constructions.

To the best of our knowledge, our constructions have sizes that are less than
the ones known from the literature.

\section{Old and New Results}\labell{Applep}

\subsection{Perfect Hash Family}\labell{PHp}\labell{Apple1p}

Let $H$ be a family of functions $h:[n]\to [q]$.
For $d\le q$ we say that $H$ is an $(n,q,d)$-{\it perfect hash family} ($(n,q,d)$-PHF)
\cite{AMS06} if for every
subset $S\subseteq [n]$ of size $|S|=d$ there is a {\it hash
function} $h\in H$ such that $h|_S$ is injective (one-to-one) on $S$, i.e.,
$|h(S)|=d$.

Blackburn and Wild \cite{BW98} gave an optimal explicit construction when
$$q\ge 2^{O\left(\sqrt{d\log d\log n}\right)}.$$ Stinson et al., \cite{SWZ}, gave an explicit construction of
$(n,q,d)$-PHF of size $d^{\log^* n}\log n$
for $q\ge d^2\log n/\log q$. It follows from the technique used in \cite{A86} with
Reed-Solomon codes that an explicit $(n,q,d)$-PHF of size $d^{2}\log n/\log q$
exist for $q\ge d^2\log n/\log q$.
In \cite{ABNNR,NSS95,AMS06} it was shown that there are
$(n,\Omega(d^2),d)$-PHF of size $O(d^6\log n)$ that can
be constructed in $poly(n)$ time. Wang and Xing \cite{WX01} used
algebraic function fields and gave an $(n,d^4,d)$-PHF of size $O((d^2/\log d)\log n)$ for infinite sequence of
integers $n$. Their construction is not linear time construction.
The above constructions are either for large~$q$ or are not linear time
constructions.

Bshouty in \cite{B14} shows that for a constant $c>1$, the following (third column in the table)
$(n,q,d)$-PHF can be
locally explicitly constructed in almost linear time (within $poly(\log)$)
{
\begin{center}
\begin{tabular}{|c|l|c|c|c|}
\hline
 &  & Linear time. & Upper & Lower\\
$n$ & $q$ & Size $=O()$ & Bound & Bound \\
\hline\hline
I.S. & $q\ge \frac{c}{4} d^4$ & $d^2\frac{\log n}{\log q}$ &$d\frac{\log n}{\log q}$& $d\frac{\log n}{\log q}$\\
\hline
all & $q\ge \frac{c}{4} d^4$ & $d^4\frac{\log n}{\log q}$ &$d\frac{\log n}{\log q}$& $d\frac{\log n}{\log q}$\\
\hline
I.S. & $ q\ge \frac{c}{2} d^2$ & $d^4\frac{\log n}{\log d}$ &$d\frac{\log n}{\log (2q/(d(d-1)))}$& $d\frac{\log n}{\log q}$\\
\hline
all & $ q\ge \frac{c}{2} d^2$ & $d^6\frac{\log n}{\log d}$ &$d\frac{\log n}{\log (2q/(d(d-1)))}$& $d\frac{\log n}{\log q}$\\
\hline
I.S. & $ q= \frac{d(d-1)}{2}+1+o(d^2)$ & $d^6\frac{\log n}{\log d}$ &$d\log n$& $d\frac{\log n}{\log q}$\\
\hline
all & $ q= \frac{d(d-1)}{2}+1+o(d^2)$ & $d^8\frac{\log n}{\log d}$ &$d\log n$& $d\frac{\log n}{\log q}$\\
\hline
\end{tabular}
\end{center}}

The upper bound in the table follows from union bound~\cite{B14}. The lower
bound is from~\cite{M82,Bl} (see also \cite{N94,FK84,K86,KM88,BW98,BT11,BESZ}). I.S. stands for ``true for infinite
sequence of integers $n$''.

Here we prove
\begin{theorem}\label{ThH1c}
Let $q$ be a power of prime. If $q>4(d(d-1)/2+1)$ then
there is a $(n,q,d)$-PHF
of size
$$O\left(\frac{d^2\log n}{\log(q/e(d(d-1)/2+1))}\right)$$
that can be constructed in linear time.

If $d(d-1)/2+2\le q\le 4(d(d-1)/2+1)$ then
there is a $(n,q,d)$-PHF
of size
$$O\left(\frac{q^2d^2\log n}{(q-d(d-1)/2-1)^2}\right)$$
that can be constructed in linear time.

In particular, for any constants $c>1$, $\delta>0$ and $0\le \eta<1$, the following $(n,q,d)$-PHF
can be constructed in linear time
(the third column in the following table)
{
\begin{center}
\begin{tabular}{|c|l|c|c|c|}
\hline
 &  & Linear time. & Upper & Lower\\
$n$ & $q$ & Size $=O()$ & Bound & Bound \\
\hline\hline
all & $q\ge  d^{2+\delta}$ & $d^2\frac{\log n}{\log q}$ &$d\frac{\log n}{\log q}$& $d\frac{\log n}{\log q}$\\
\hline
all & $ q\ge \frac{c}{2}d^2$ & $d^2{\log n}$ &$d\log n$& $d\frac{\log n}{\log q}$\\
\hline
all & $ q= \frac{d(d-1)}{2}+1+d^{2\eta}$ & $d^{6-4\eta}{\log n}$ &$d\log n$& $d\frac{\log n}{\log q}$\\
\hline
all & $ q= \frac{d(d-1)}{2}+2$ & $d^6\frac{\log n}{\log d}$ &$d\log n$& $d\frac{\log n}{\log q}$\\
\hline
\end{tabular}
\end{center}}
\end{theorem}
Notice that for $q>cd^2/2$, $c>1$ the sizes
in the above theorem is within a factor of $d$ of the lower bound.
Constructing almost optimal (within $poly(d)$) $(n,q,d)$-PHF
for $q= o(d^2)$ is still a challenging open problem.
Some nearly optimal constructions of $(n,q,d)$-PHF for $q= o(d^2)$ are given in~\cite{NSS95,M03}.

The $(n,q,d)$-perfect hash families for $d\le 6$ are studied in
\cite{AMSW96,BW98,Bl,SWZ,M03,BJQ04,BJ08,MT08}. In this paper we prove
\begin{theorem} If $q$ is prime power and $d\le \log n/(8\log\log n)$ then there is a linear time
construction of $(n,q,d)$-PHF of size
$$O\left(\frac{d^3 \log n}{g(q,d)}\right)$$
where
$$g(q,d)=\left(1-\frac{1}{q}\right)\left(1-\frac{2}{q}\right)\cdots \left(1-\frac{d-1}{q}\right).$$
\end{theorem}
Using the lower bound in \cite{FK84} we show that the size
in the above theorem is within a factor of $d^4$ of the lower bound when
$q=d+O(1)$ and within a factor of $d^3$ for $q>cd$ for some $c>1$.

\subsection{Dense Perfect Hash Family}

We say that $H$ is an $(1-\epsilon)$-{\it dense} $(n,q,d)$-PHF
if for every
subset $S\subseteq [n]$ of size $|S|=d$ there are at least
$(1-\epsilon)|H|$ hash
functions $h\in H$ such that $h|_S$ is injective on $S$.

We prove
\begin{theorem}~\label{Den} Let $q$ be a power of prime. If $\epsilon>4(d(d-1)/2+1)/q$ then
there is a $(1-\epsilon)$-dense $(n,q,d)$-PHF
of size
$$O\left(\frac{d^2\log n}{\epsilon\log(\epsilon q/e(d(d-1)/2+1))}\right)$$
that can be constructed in linear time.

If $(d(d-1)/2+1)/(q-1)\le \epsilon \le 4(d(d-1)/2+1)/q$ then
there is a $(1-\epsilon)$-dense $(n,q,d)$-PHF
of size
$$O\left(\frac{q^2d^2\log n}{\epsilon(q-(d(d-1)/2+1)/\epsilon)^2}\right)$$
that can be constructed in linear time.
\end{theorem}

We also prove (what we believe) two folklore results
that show that the bounds on the size and $\epsilon$ in the above theorem are almost tight.
First, we show that the size  of any $(1-\epsilon)$-dense $(n,q,d)$-PHF is
$$\Omega\left(\frac{d\log n}{\epsilon \log q}\right).$$
Second, we show that no $(1-\epsilon)$-dense $(n,q,d)$-PHF
exists when $\epsilon<d(d-1)/(2q)+O((d^2/q)^2)$.

Notice that for $q\ge (d/\epsilon)^{1+c}$, where $c>1$ is any constant, the
size of the construction in Theorem~\ref{Den},
$$O\left(\frac{d^2\log n}{\epsilon \log q}\right),$$ is within a factor $d$ of the
lower bound. Also the bound on $\epsilon$ is asymptotically tight.

For the rest of this section we will only state the results
for the non-dense $d$-restriction problems. Results similar to
Theorem~\ref{Den} can be easily obtained using the same technique.

\subsection{Cover-Free Families}\labell{CFF}\labell{Apple3}

Let $X$ be a set with $N$ elements and let $\cB$ be a set of subsets (blocks) of $X$.
We say that $(X,\cB)$ is $(w,r)$-{\it cover-free family}
($(w,r)$-CFF),~\cite{KS64}, if for any~$w$ blocks $B_1,\ldots,B_w\in \cB$ and any other $r$ blocks $A_1,\ldots,A_r\in \cB$, we have
$$\bigcap_{i=1}^w B_i\not\subseteq \bigcup_{j=1}^r A_j.$$
Let $N((w,r),n)$ denotes the minimum number of points in any $(w,r)$-CFF having $n$ blocks.
Here we will study CFF when $w=o(r)$ (or $r=o(w)$). We will write $(n,(w,r))$-CFF
when we want to emphasize the number of blocks.

When $w=1$, the problem is called {\it group testing}.
The problem of group testing which was first presented during
World War II was presented as follows~\cite{DH00,ND00}: Among $n$ soldiers, at most
$r$ carry a fatal virus. We would like to blood test the soldiers
to detect the infected ones. Testing each one separately will give
$n$ tests. To minimize the number of tests we can mix the blood of
several soldiers and test the mixture. If the test comes negative
then none of the tested soldiers are infected. If the test comes
out positive, we know that at least one of them is infected. The
problem is to come up with a small number of tests.

This problem is equivalent to $(n,(1,r))$-CFF and is equivalent to finding a small set $\cF\subseteq\{0,1\}^n$ such that for every $1\le i_1<i_2<\cdots<i_d\le
n$, $d=r+1$, and every $1\le j\le d$ there is $a\in \cF$ such that $a_{i_k}=0$
for all $k\not=j$ and $a_{i_j}=1$.

Group testing has the following lower bound \cite{DR82,DR89,F96}
\begin{eqnarray}\labell{CFFlower}
N((1,r),n)\ge\Omega\left( \frac{r^2}{\log r}\log n\right).
\end{eqnarray}
It is known that a group testing of size $O(r^2\log
n)$ can be constructed in linear time \cite{DH00,PR11,INR10}.

An $(n,(w,r))$-CFF
can be regarded as a set $\cF\subseteq \{0,1\}^n$ such that for every $1\le i_1<i_2<\cdots<i_{d}\le n$ where $d=w+r$ and
every $J\subset [d]$ of size $|J|=w$ there is $a\in \cF$
such that $a_{i_k}=0$ for all $k\not\in J$ and $a_{i_j}=1$ for all
$j\in J$. Then $N((w,r),n)$ is the minimum size of such $\cF$.

It is known that, \cite{SWZ00},
$$N((w,r),n)\ge \Omega\left(\frac{d{d\choose w}}{\log{d\choose w}}\log n\right).$$

Using union bound it is easy to show
\begin{lemma} For $d=w+r=o(n)$ we have
$$N((w,r),n)\le O\left(\sqrt{{wrd}}\cdot {d\choose w}\log n\right).$$
\end{lemma}

It follows from \cite{SWZ}, that for infinite sequence of integers $n$, an $(n,(w,r))$-CFF of size
$$M=O\left((wr)^{\log^* n}\log n\right)$$ can be constructed in polynomial time.
For constant $d$, the $(n,d)$-universal set over $\Sigma=\{0,1\}$ constructed in \cite{NN93} of size $M=O(2^{3d}\log n)$
(and in \cite{NSS95} of size $M=2^{d+O(\log^2d)}\log n$) is $(n,(w,r))$-CFF for any $w$ and $r$ of size $O(\log n)$. See also \cite{LS06}.
In \cite{B14}, Bshouty gave the following locally explicit constructions
of $(n,(w,r))$-CFF that can be constructed in (almost) linear time in their sizes (the third column in the table).
\begin{center}
\begin{tabular}{|c|c|c|c|c|}
\hline
 &  & Linear time  & Upper & Lower \\

$n$ & $w$ & Size= & Bound & Bound \\
\hline\hline
I.S & $O(1)$ & $\frac{r^{w+2}}{\log r}\log n$ & $r^{w+1}\log n$ & $\frac{r^{w+1}}{\log r}\log n$\\
\hline
all & $O(1)$ & $\frac{r^{w+3}}{\log r}\log n$ & $r^{w+1}\log n$ & $\frac{r^{w+1}}{\log r}\log n$\\
\hline
I.S. & $o(r)$ & $\frac{w^2(ce)^wr^{w+2}}{\log r}\log n$ & $\frac{r^{w+1}}{(w/e)^{w-1/2}}\log n$ & $\frac{r^{w+1}}{(w/e)^{w+1}\log r}\log n$\\
\hline
all & $o(r)$ & $\frac{w^3(ce)^wr^{w+3}}{\log r}\log n$ & $\frac{r^{w+1}}{(w/e)^{w-1/2}}\log n$ & $\frac{r^{w+1}}{(w/e)^{w+1}\log r}\log n$\\
\hline
\end{tabular}
\end{center}
In the table, $c>1$ is any constant.
We also added to the table the non-constructive upper bound in the forth column
and the lower bound in the fifth column.

In this paper we prove
\begin{theorem} \labell{CFFrp} For any constant $c>1$, the following $(n,(w,r))$-CFF can be constructed in linear time in their sizes
\begin{center}
\begin{tabular}{|c|c|c|c|c|}
\hline
 &  & Linear time.  & Upper & Lower \\

$n$ & $w$ & Size=$O(\ )$ & Bound & Bound \\
\hline\hline
all & $O(1)$ & ${r^{w+1}}\log n$ & $r^{w+1}\log n$ & $\frac{r^{w+1}}{\log r}\log n$\\
\hline
all & $o(r)$ & ${(ce)^wr^{w+1}}\log n$ & $\frac{r^{w+1}}{(w/e)^{w-1/2}}\log n$ & $\frac{r^{w+1}}{(w/e)^{w+1}\log r}\log n$\\
\hline
\end{tabular}
\end{center}
\end{theorem}
Notice that when $w=O(1)$ the size of the construction matches the upper bound obtained with union bound
and is within a factor of $\log r$ of the lower bound.

\subsection{Separating Hash Family}\labell{SHF}\labell{Apple4}

Let $X$ and $\Sigma$ be sets of cardinalities $n$ and $q$, respectively. We call a set $\cF$ of
functions $f:X\to \Sigma$ an $(M; n,q,(d_1,d_2,\ldots,d_r))$-{\it separating hash family} (SHF),~\cite{STW00,SWC08},
if $|\cF|=M$ and for all pairwise disjoint subsets $C_1,C_2,\ldots,C_r\subseteq X$ with $|C_i|=d_i$ for $i=1,2,\ldots,r$, there is at
least one function $f\in \cF$ such that $f(C_1),f(C_2),\ldots,f(C_r)$ are pairwise disjoint subsets. The goal is to find
$(M; n,q,(d_1,d_2,\ldots,d_r))$-SHF with small $M$. The minimal $M$ is denoted by $M(n,q,(d_1,d_2,\ldots,d_r))$.

Notice that $(n,q,d)$-PHF of size $M$ is $(M;n,q,(1,1,\stackrel{d}{\ldots},1))$-SHF
and $(w,r)$-CFF of size $M$ is $(M;n,2,(r,w))$-SHF.

In \cite{BT11}, Bazrafshan and Trund proved that for
$$D_1=\sum_{i=1}^r d_i,$$
\begin{eqnarray}\labell{black}
M(n,q,(d_1,d_2,\ldots,d_r))&\ge & \left(D_1-1\right)\frac{\log n-\log(D_1-1)-\log q}{\log q}\nonumber\\
&=&\Omega\left(D_1\frac{\log n}{\log q}\right).\end{eqnarray} See also \cite{BESZ}.

In \cite{SWZ}, Stinson et. al. proved that an $(M; n,q,(d_1,d_2))$ separating hash families of size
$$M=O((d_1d_2)^{\log^* n}\log n)$$ can be constructed in polynomial time for infinite sequence of integers $n$ and $q>d_1d_2$. The same proof gives a polynomial time construction for any separating hash family of size
$$M=O(D_2^{\log^*n}\log n)$$ where $$D_2=\sum_{1\le i_1<i_2\le r}d_{i_1}d_{i_2}$$ when $q>D_2$.

In \cite{LS06}, Liu and Shen provide an explicit constructions of $(M; n,q,$ $(d_1,d_2))$ separating hash families using algebraic curves over finite fields. They show that for infinite sequence of integers $n$ there is an explicit $(M; n,q,(d_1,d_2))$ separating hash families of size $O(\log n)$ for fixed $d_1$ and $d_2$. This also follows from~\cite{NN93}, an $(n,d_1+d_2)$-universal set over two symbols alphabet is a separating hash families of size $O(\log n)$ for fixed $d_1$ and $d_2$. Their construction is similar to the construction of the tester in \cite{B14}.
In \cite{B14} Bshouty gave a polynomial time construction of an $(M; n,q,(d_1,d_2))$ separating hash families of size $M=((d_1d_2)^4 \log n/\log q)$ for any $q\ge d_1d_2(1+o(1))$ and any~$n$.
He also show that for any constant $c>1$ and $q>D_2$, the following $(M; n,q,(d_1,d_2,\ldots,d_r))$  separating hash family can be constructed in polynomial time
{
\begin{center}
\begin{tabular}{|c|l|c|c|c|}
\hline
 &  & poly time. & Upper & Lower\\
$n$ & $q$ & Size $M=O(\ )$ & Bound & Bound \\
\hline\hline
I.S. & $q\ge {c} (D_2+1)^2$, $q$ P.S. & $D_2\frac{\log n}{\log q}$ &$D_1\frac{\log n}{\log q}$& $D_1\frac{\log n}{\log q}$\\
\hline
all & $q\ge {c} (D_2+1)^2$, $q$ P.S.& $D_2^2\frac{\log n}{\log q}$ &$D_1\frac{\log n}{\log q}$& $D_1\frac{\log n}{\log q}$\\
\hline
I.S. & $ q\ge {c} (D_2+1)$ & $D_2^2\frac{\log n}{\log D_2}$ &$D_1{\log n}$& $D_1\frac{\log n}{\log q}$\\
\hline
all & $ q\ge {c} (D_2+1)$ & $D_2^3\frac{\log n}{\log D_2}$ &$D_1{\log n}$& $D_1\frac{\log n}{\log q}$\\
\hline
I.S. & $ q\ge D_2+1$ & $D_2^3\frac{\log n}{\log D_2}$ &$D_1\log n$& $D_1\frac{\log n}{\log q}$\\
\hline
 all & $ q\ge D_2+1$ & $D_2^4\frac{\log n}{\log D_2}$ &$D_1\log n$& $D_1\frac{\log n}{\log q}$\\
\hline
\end{tabular}
\end{center}}
and an $(M; n,r,(d_1,d_2,\ldots,d_r))$  separating hash family of size
$$M=\frac{{cD_2\choose d_1\ d_2\ \cdots\ d_r} D_2^3}{\log D_2} \log n,$$
can be constructed in time linear in the construction size.

Here we prove the following
\begin{theorem} \labell{SHFrespp} For any constant $c>1$ and $q>D_2$, the following $(M; n,q$ $,(d_1,d_2,\ldots,d_r))$  separating hash family can be constructed in linear time
{
\begin{center}
\begin{tabular}{|c|l|c|c|c|}
\hline
 &  & poly time. & Upper & Lower\\
$n$ & $q$ & Size $M=O(\ )$ & Bound & Bound \\
\hline\hline
all & $ q\ge (D_2+1)^{c}$ & $D_2\frac{\log n}{\log q}$ &$D_1\frac{\log n}{\log q}$& $D_1\frac{\log n}{\log q}$\\
\hline
all & $ q\ge {c} (D_2+1)$ & $D_2{\log n}$ &$D_1{\log n}$& $D_1\frac{\log n}{\log q}$\\
\hline
 all & $ q\ge D_2+2$ & $D_2^3{\log n}$ &$D_1\log n$& $D_1\frac{\log n}{\log q}$\\
\hline
\end{tabular}
\end{center}}
and an $(M; n,r,(d_1,d_2,\ldots,d_r))$  separating hash family of size
$$M={{cD_2\choose d_1\ d_2\ \cdots\ d_r} D_2}\log n,$$
can be constructed in time linear in the construction size.
\end{theorem}

\section{Preliminary Constructions}
A {\it linear code} over the field $\FF_q$ is a linear
subspace $C\subset \FF_q^m$. Elements in the code
are called {\it words}. A linear code $C$ is called
$[m,k,d]_q$ {\it linear code} if $C\subset \FF_q^m$
is a linear code, $|C|=q^k$ and for every two words $v$ and $u$
in the code $\dist(v,u):=|\{i\ |\ v_i\not=u_i\}|\ge d$.

The $q$-ary entropy function is
$$H_q(p)=p\log_q \frac{q-1}{p}+(1-p)\log_q\frac{1}{1-p}.$$

The following is from \cite{PR11} (Theorem 2)
\begin{lemma} \label{codeo} Let $q$ be a prime power, $m$ and $k$ positive integers
and $0\le \delta\le 1$. If $k\le (1-H_q(\delta))m$, then an $[m,k,\delta m]_q$ linear code
can be globally explicit constructed in time $O(mq^k)$.
\end{lemma}
Notice that to construct $n$ codewords in an $[m,k,\delta m]_q$ linear code where $q^{k-1}<n\le q^k$,
the time of the construction is $O(mq^k)=O(qmn)$.

We now show
\begin{lemma} \label{code} Let $q$ be a prime power, $m$ and $k$ positive integers,
 $0\le \delta\le 1$ and $n$ an integer such that $q^{k-1}<n\le q^k$.
 If $k\le (1-H_q(\delta))m$, then a set of
$n$ codewords in an $[m,k,\delta m]_q$ linear code
can be globally explicit constructed in time $O(mn)$.
\end{lemma}
\begin{proof} The same proof of Lemma~\ref{code} in~\cite{PR11}
works here with the observation that the first column of the generator
matrix $G=[v_1|\cdots|v_k]$ can be the all-one vector $v_1=(1,1,\ldots,1)'$ and
it is enough to ensure that every codeword of the form $\sum_{i=1}^k \lambda_i v_i$,
where the first nonzero $\lambda_i$ is $1$, is of weight at least $\delta m$.
We call such codeword a normalized codeword. The number of normalized codewords
is $(q^k-1)/(q-1)=O(q^{k-1})$. Obviously, codeword $v$ is equal to $\lambda u$ for
some normalized codewords $u$ and the minimum weight of normalized codeword
is the minimum weight of the code.

Therefore we first construct the normalized codewords as in \cite{PR11}
in time $O(mq^{k-1})$
and then add any $n-(q^{k}-1)/(q-1)$ codewords.
\end{proof}

All the results in this paper uses Lemma~\ref{code} and
therefore they are globally explicit constructions.
We now show
\begin{lemma}\label{spcode} Let $q$ be a prime power, $1<h<q/4$ and
$$m=\left\lceil \frac{h\ln (q(n+1))}{\ln q-\ln h-1}\right\rceil.$$
A set of $n$ nonzero codewords of a $$\left[m,\left\lceil\frac{\log (n+1)}{\log q}\right\rceil,\left(1-\frac{1}{h}\right) m\right]_q$$ linear code
can be constructed in time $O(nm)$.
\end{lemma}
\begin{proof} By Lemma~\ref{code} it is enough to show that
$$\left\lceil\frac{\log (n+1)}{\log q}\right\rceil\le \left( 1-H_q\left(1-\frac{1}{h}\right)\right)m.$$
Now since for $x>0$, $(x-1)/x\le \ln x$ we have
\begin{eqnarray*}
1-H_q\left(1-\frac{1}{h}\right)&=& 1-\left(\left(1-\frac{1}{h}\right)\log_q \frac{q-1}{1-1/h}+\frac{1}{h}\log_qh\right)\\
&=& \frac{1}{h}-\frac{1}{h}\log_qh- \left(1-\frac{1}{h}\right)\log_q \frac{1-1/q}{1-1/h}\\
&\ge& \frac{1}{h}-\frac{1}{h}\log_qh+ \frac{h-1}{h}\log_q ({1-1/h})\\
&\ge& \frac{1}{h}-\frac{1}{h}\log_qh- \frac{1}{h\ln q}\\
&=& \frac{\ln q-\ln h-1}{h\ln q}.
\end{eqnarray*}
Now
$$\left( 1-H_q\left(1-\frac{1}{h}\right)\right)m\ge \frac{\ln q(n+1)}{\ln q}\ge \left\lceil\frac{\log (n+1)}{\log q}\right\rceil.$$
\end{proof}

When $h=O(q)$ we show
\begin{lemma}\label{spcode2} Let $q$ be a prime power, $2\le q/4\le h\le q-1$ and
$$m=\left\lceil \frac{4(q-1)^2h\ln (q(n+1))}{(q-h)^2}\right\rceil.$$
A set of $n$ nonzero codewords of a $$\left[m,\left\lceil\frac{\log (n+1)}{\log q}\right\rceil,\left(1-\frac{1}{h}\right) m\right]_q$$ linear code
can be constructed in time $O(nm)$.
\end{lemma}
\begin{proof} For $\Delta=1-H_q\left(1-\frac{1}{h}\right)$ and using the fact that
$\ln(1-x)=-x-x^2/2-x^3/3-\cdots$ for $|x|<1$, we have
\begin{eqnarray*}
\Delta&=&  \frac{1}{h}-\frac{1}{h}\log_qh- \left(1-\frac{1}{h}\right)\log_q \frac{1-1/q}{1-1/h}\\
&=& \frac{1}{\ln q}\left( \frac{1}{h}\ln\frac{q}{h}-\left(1-\frac{1}{h}\right) \ln \frac{h(q-1)}{q(h-1)}\right)\\
&=& \frac{1}{\ln q}\left( -\frac{1}{h}\ln\left(1-\frac{q-h}{q-1}\right)+\ln \left(1-\frac{q-h}{h(q-1)}\right)\right)\\
&=& \frac{(q-h)^2}{(q-1)^2h\ln q}\left(\frac{1}{2}\left(1-\frac{1}{h}\right) +\frac{q-h}{3(q-1)}\left(1-\frac{1}{h^2}\right)+\cdots\right)\\
&\ge& \frac{(q-h)^2}{(q-1)^2h\ln q}\left(\frac{1}{2}\left(1-\frac{1}{h}\right)\right)\ge \frac{(q-h)^2}{4(q-1)^2h\ln q}.
\end{eqnarray*}
Now
$$\left( 1-H_q\left(1-\frac{1}{h}\right)\right)m\ge \frac{\ln q(n+1)}{\ln q}\ge \left\lceil\frac{\log (n+1)}{\log q}\right\rceil.$$
\end{proof}

\section{Main Results}
In this section we give two main results that will be used throughout the paper

Let $I\subseteq [n]^2$. Define the following homogeneous polynomial
$$H_I=\prod_{(i_1,i_2)\in I}(x_{i_1}- x_{i_2}).$$
We denote by ${\cal H}_d\subseteq \FF_q[x_1,\ldots,x_n]$ the class of all such polynomials of degree at most~$d$. A {\it hitting set}
for ${\cal H}_d$ over $\FF_q$ is a set of assignment $A\subseteq \FF_q^n$ such that for every $H\in {\cal H}_d, H\not\equiv 0$,
there is $a\in A$ where $H(a)\not=0$. A $(1-\epsilon)$-{\it dense hitting set}
for ${\cal H}_d$ over $\FF_q$ is a set of assignment $A\subseteq \FF_q^n$ such that for every $H\in {\cal H}_d$, $H\not\equiv 0$,
$$\Pr_{a\in A}[H(a)\not=0]> 1-\epsilon$$ where the probability is over the choice
of $a$ from the uniform distribution on $A$. When $H(a)\not=0$ then we say that the assignment
$a$ {\it hits} $H$ and $H$ is {\it not zero on} $a$.

We prove
\begin{lemma} \label{hB} Let $n>q,d$. If $q>4(d+1)$ is prime power
then there is a hitting set for ${\cal H}_d$ of size
$$m=\left\lceil \frac{(d+1)\log(q(n+1))}{\log (q/e(d+1))}\right\rceil=O\left(\frac{d\log n}{\log(q/e(d+1))}\right)$$
that can be constructed in time $O(mn)=O(dqn\log(qn))$.

If $d+2\le q\le 4(d+1)$ is prime power then there is a hitting set for ${\cal H}_d$ of size
$$m=\left\lceil \frac{4(q-1)^2(d+1)\ln(q(n+1))}{(q-d-1)^2}\right\rceil=O\left(\frac{dq^2\log n}{(q-d-1)^2}\right)$$
that can be constructed in time $O(mn)=O(d(q^2/(q-d-1)^2)n\log(qn))$.
\end{lemma}
\begin{proof} Consider the code $C$
$$\left[m,\left\lceil\frac{\log (n+1)}{\log q}\right\rceil,\left(1-\frac{1}{d+1}\right) m\right]_q$$
constructed in Lemma~\ref{spcode} and Lemma~\ref{spcode2}.
The number of non-zero words in the code is at least $n$. Take any $n$ distinct non-zero words
$c^{(1)},\cdots,c^{(n)}$ in $C$ and define the assignments $a^{(i)}\in \FF_q^n$, $i=1,\ldots,m$ where
$a^{(i)}_j=c^{(j)}_i$.
Let $H_I\in {\cal H}_d, H_I\not\equiv 0$. Then
$$H_I=\prod_{(i_1,i_2)\in I}(x_{i_1}-x_{i_2})\not\equiv 0$$
where $|I|\le d$. For each $t:=x_{i_1}- x_{i_2}$ we have
$(t(a^{(1)}),\ldots,t(a^{(m)}))^T=c^{(i_1)}- c^{(i_2)}\in C$ is a non-zero word in $C$ and therefore
$t$ is zero on at most $m/(d+1)$ assignments. Therefore $H_I$ is
zero on at most $dm/(d+1)<m$ assignment. This implies that there is an assignment in $A$ that hits $H_I$.
\end{proof}
Notice that the size of the hitting set is $mn$ and therefore the
time complexity in the above lemma is linear in the size of the hitting set.

In the same way one can prove
\begin{lemma} \label{hBD} Let $q$ be a prime power.
If $q> 4(d+1)/\epsilon$ be a prime power.
Let $n>q,d$. There is a $(1-\epsilon)$-dense hitting set for ${\cal H}_d$ of size
$$m=\left\lceil \frac{(d+1)\log(q(n+1))}{\epsilon \log (\epsilon q/e(d+1))}\right\rceil=O\left(\frac{d\log n}{\epsilon\log(\epsilon q/e(d+1))}\right)$$
that can be constructed in time $O(dqn\log(qn)/\epsilon)$.

If $(d+1)/\epsilon+1\le q\le 4(d+1)/\epsilon$ be a prime power.
Let $n>q,d$. There is a $(1-\epsilon)$-dense hitting set for ${\cal H}_d$ of size
$$m=\left\lceil \frac{4(q-1)^2(d+1)\ln(q(n+1))}{(q-(d+1)/\epsilon)^2\epsilon}\right\rceil=O\left(\frac{dq^2\log n}{(q-(d+1)/\epsilon)^2\epsilon}\right)$$
that can be constructed in time $O(d(q^2/(q-d-1)^2)n\log(qn)/\epsilon)$.
\end{lemma}
We note here that such result cannot be achieved when $q<d/\epsilon$~\cite{B14}.

\section{Proof of the Theorems}\labell{Apple}

\subsection{Perfect Hash Family}\labell{PH}\labell{Apple1}

Here we prove
\setcounter{theorem}{0}
\begin{theorem}\label{ThH1} Let $q$ be a power of prime. If $q>4(d(d-1)/2+1)$ then
there is a $(n,q,d)$-PHF
of size
$$O\left(\frac{d^2\log n}{\log(q/e(d(d-1)/2+1))}\right)$$
that can be constructed in linear time.

If $d(d-1)/2+2\le q\le 4(d(d-1)/2+1)$ then
there is a $(n,q,d)$-PHF
of size
$$O\left(\frac{q^2d^2\log n}{(q-d(d-1)/2-1)^2}\right)$$
that can be constructed in linear time.

In particular,
for any constants $c>1$, $\delta>0$ and $0\le \eta<1$, the following $(n,q,d)$-PHF can be constructed in linear time
(the third column in the following table)
{
\begin{center}
\begin{tabular}{|c|l|c|c|c|}
\hline
 &  & Linear time. & Upper & Lower\\
$n$ & $q$ & Size $=O()$ & Bound & Bound \\
\hline\hline
all & $q\ge  d^{2+\delta}$ & $d^2\frac{\log n}{\log q}$ &$d\frac{\log n}{\log q}$& $d\frac{\log n}{\log q}$\\
\hline
all & $ q\ge \frac{c}{2}d^2$ & $d^2{\log n}$ &$d\log n$& $d\frac{\log n}{\log q}$\\
\hline
all & $ q= \frac{d(d-1)}{2}+1+d^{2\eta}$ & $d^{6-4\eta}{\log n}$ &$d\log n$& $d\frac{\log n}{\log q}$\\
\hline
all & $ q= \frac{d(d-1)}{2}+2$ & $d^6\frac{\log n}{\log d}$ &$d\log n$& $d\frac{\log n}{\log q}$\\
\hline
\end{tabular}
\end{center}}
\end{theorem}
\begin{proof} Consider the set of functions
$${\cal F}=\{\Delta_{\{i_1,\ldots,i_d\}}(x_{1},\ldots,x_{n})\ |\ 1\le
i_1<\cdots<i_d\le n\}$$ in $\FF_q[x_1,x_2,\ldots,x_n]$ where
$$\Delta_{\{i_1,\ldots,i_d\}}(x_1,\ldots,x_n)=
\prod_{1\le k<j\le d}(x_{i_k}-x_{i_j}).$$ It is clear that a hitting set for $\cF$
is $(n,q,d)$-PHF.
Now since $\cF\subseteq {\cal H}_{d(d-1)/2+1}$ the result follows from
Lemma~\ref{hB}.
\end{proof}

When $q>d(d-1)/2$ is not a power of prime number then we can take
the nearest prime $q'<q$ and construct an $(n,q',d)$-PHF that is also $(n,q,d)$-PHF. It is known that
the nearest prime $q'\ge q-\Theta(q^{.525})$,~\cite{BHP01}, and
therefore the result in the above table is also true for any
integer $q\ge d(d+1)/2+O(d^{1.05})$.

\subsection{Perfect Hash Family for Small $d$}\label{The}
We now prove
\begin{theorem} If $q$ is prime power and $d\le \log n/(8\log\log n)$ then there is a linear time
construction of $(n,q,d)$-PHF of size
$$O\left(\frac{d^3 \log n}{g(q,d)}\right)$$
where
$$g(q,d)=\left(1-\frac{1}{q}\right)\left(1-\frac{2}{q}\right)\cdots \left(1-\frac{d-1}{q}\right).$$
\end{theorem}
\begin{proof} If $q>d^2$ then the construction in Theorem~\ref{ThH1}
has the required size.
Let $q\le d^2$. We first use Theorem~\ref{ThH1}
to construct an $(n,d^3,d)$-PHF $H_1$ of size
$O(d^2\log n/\log d)$ in linear time. Then a $(d^3,q,d)$-PHF $H_2$ of size $O(d\log d/g(q,d))$ can be constructed in time, \cite{NSS95,AMS06},
$${d^3\choose d}q^{1+\lceil \log d^3/\log q\rceil(d-1)}\le d^{3d}q^dd^{3d}\le d^{8d}<n.$$
Then $H=\{h_2(h_1)\ |\ h_2\in H_2,h_1\in H_1\}$
is $(n,q,d)$-PHF of the required size.
\end{proof}

We now show that this bound is within a factor of $d^4$ of the lower bound when
$q=d+O(1)$ and within a factor of $d^3\log d$ of the lower bound when $q>cd$ for some constant $c>1$.

\begin{lemma} \cite{FK84} Let $n>d^{2+\epsilon}$ for some constant $\epsilon>0$.
Any $(n,q,d)$-PHF is of size at least
$$\Omega\left(\frac{(q-d+1)}{q\log(q-d+2)} \frac{\log n}{g(q,d)}\right).$$

In particular, for $q=d+O(1)$ the bound is
$$\Omega\left(\frac{\log n}{d g(q,d)}\right)$$
and for $q>cd$ for some constant $c>1$ the bound is
$$\Omega\left(\frac{\log n}{(\log d) g(q,d)}\right).$$
\end{lemma}
\ignore{\begin{proof} Let $H$ be a $(n,q,d)$-perfect hash family.
Let $1\le i_1<i_2<\cdots<i_{d-1}\le n$ and consider
$H'=\{h\ |\ h(i_1),\ldots,h(i_{d-2})$ are distinct$\}$.
For every $h\in H'$ we define a permutation $g_h:[q]\to [q]$ that
maps $\{h(i_1),\ldots, h(i_{d-2})\}$ to $\{q-d+3,\ldots,q\}$.
It is easy to see that $H'$ is $(n-d+2,q-d+2,2)$-perfect hash family
(over $[n]\backslash \{i_1,\ldots,i_{d-2}\}$) and therefore
$$|H'|\ge t:=\frac{\log (n-d+2)}{\log (q-d+2)}.$$ This shows that
for every $1\le i_1<i_2<\cdots<i_{d-2}\le n$ there is at least~$t$
functions in $H$ that is perfect for $S=\{i_1,\ldots,i_{d-2}\}$.
Since each hash function can be perfect for at most ${q\choose d-2}(n/q)^{d-2}$
and we need to cover each of the ${n\choose d-2}$ subsets, $R\subseteq [n]$, $t$ times, we must have
(for some constant $c$)
\begin{eqnarray*}
|H|&\ge& \frac{{n\choose d-2}t}{{q\choose d-2}(n/q)^{d-2}}\\
&\ge&c \frac{\log n}{g(q,d-2)\log (q-d+2)}\\
&=&\Omega\left(\frac{(q-d+1)(q-d+2)}{q^2\log(q-d+2)} \frac{\log n}{g(q,d)}\right).
\end{eqnarray*}
\end{proof} }

\subsection{Dense Perfect Hash}\labell{PH2}\labell{Apple12}

Using Lemma~\ref{hBD} with the same proof as in Theorem~\ref{ThH1} we get
\begin{theorem} Let $q$ be a power of prime. If $q>4(d(d-1)/2+1)/\epsilon$ then
there is a $(1-\epsilon)$-dense $(n,q,d)$-{\it perfect hash family}
of size
$$O\left(\frac{d^2\log n}{\epsilon\log(\epsilon q/e(d(d-1)/2+1))}\right)$$
that can be constructed in linear time.

If $(d(d-1)/2+1)/\epsilon+1\le q\le 4(d(d-1)/2+1)/\epsilon$ then
there is a $(1-\epsilon)$-dense $(n,q,d)$-PHF
of size
$$O\left(\frac{q^2d\log n}{\epsilon(q-(d(d-1)/2+1)/\epsilon)^2}\right).$$
that can be constructed in linear time.
\end{theorem}

The following two folklore results are proved for completeness
\begin{lemma} Let $q\ge d^{1+c}$ for some constant $c>1$.
Any $(1-\epsilon)$-dense $(n,q,d)$-PHF
is of size at least
$$\Omega\left(\frac{d\log n}{\epsilon \log q}\right).$$
\end{lemma}
\begin{proof} If $H$ is an $(1-\epsilon)$-dense $(n,q,d)$-PHF
then any subset of $H$ of size $\epsilon |H|+1$ is $(n,q,d)$-PHF.
Now the result follows from the lower bound for the size
of $(n,q,d)$-PHF.
\end{proof}

\begin{lemma} Let $q>d^2/2$. When $$\epsilon\le \frac{d(d-1)}{2q}-\frac{d^2(d-1)^2}{8q^2}$$
then no $(1-\epsilon)$-dense $(n,q,d)$-PHF exists.
\end{lemma}
\begin{proof} Each hash function $h:[n]\to [q]$ can be perfect
for at most ${q\choose d}(n/q)^d$ sets $S$ of size $d$, \cite{FK84}. There
are exactly ${n\choose d}$ sets and therefore the density cannot be greater
than
$$1-\epsilon \le \frac{{q\choose d}\left(\frac{n}{q}\right)^d}{{n\choose d}}\stackrel{n\to\infty}{\longrightarrow}\left(1-\frac{1}{q}\right)\cdots\left(
1-\frac{d-1}{q}\right)\le e^{-d(d-1)/2q} .$$
Since $$e^{-d(d-1)/2q}\le 1-\frac{d(d-1)}{2q}+\frac{d^2(d-1)^2}{8q^2}.$$ the result follows.
\end{proof}

For the rest of the paper we will only state the results
for the non-dense $d$-restriction problems. The results
for the dense $d$-restrict problems follows immediately from applying
Lemma~\ref{hBD}.

\subsection{Cover-Free Families}\labell{CFF}\labell{Apple3}

We now prove the following
\begin{theorem} \labell{CFFr} Let $q\ge wr+2$ be a prime power.
Let $S\subseteq \FF_q^{n}$ be a hitting set for ${\cal H}_{wr}$. Given
a $(q,(w,r))$-CFF of size $M$ that can be constructed in linear time one can
construct an $(n,(w,r))$-CFF of size $M\cdot |S|$ that can be constructed in linear time.

In particular, there is an
$(w,r)$-CFF of size
$${q\choose w} \cdot |S|$$ that can be constructed in linear time in its size.

In particular, for any constant $c>1$, the following $(w,r)$-CFF can be constructed in linear time in their sizes
\begin{center}
\begin{tabular}{|c|c|c|c|c|}
\hline
 &  & Linear time.  & Upper & Lower \\

$n$ & $w$ & Size=$O(\ )$ & Bound & Bound \\
\hline\hline
all & $O(1)$ & ${r^{w+1}}\log n$ & $r^{w+1}\log n$ & $\frac{r^{w+1}}{\log r}\log n$\\
\hline
all & $o(r)$ & ${(ce)^wr^{w+1}}\log n$ & $\frac{r^{w+1}}{(w/e)^{w-1/2}}\log n$ & $\frac{r^{w+1}}{(w/e)^{w+1}\log r}\log n$\\
\hline
\end{tabular}
\end{center}
\end{theorem}
\begin{proof} Consider the set of non-zero functions $$\cM=\{\Delta_\bfi\ |\ \bfi\in[n]^d,  \ i_1,i_2,\ldots,i_d \mbox{\ are distinct}\}$$ where
$$\Delta_{\bfi}(x_1,\ldots,x_n)=\prod_{1\le k\le
w\mbox{{\small \ and\ }} w<j\le d} (x_{i_k}-x_{i_j}).$$
Then $S$ is a hitting set for ${\cal M}$.

Let $\cF\subseteq \{0,1\}^q$ be a $(q,(w,r))$-CFF of size $M$. Regard each $f\in \cF$ as a function $f:\FF_q\to\{0,1\}$. It is easy to see that
$$\{(f(b_1),f(b_2),\ldots,f(b_n))\ |\ b\in S, f\in\cF\}\subseteq \{0,1\}^n$$ is $(w,r)$-CFF of size $|\cF|\cdot|S|=M\cdot |S|$.

Now for every subset $R\subseteq \FF_q$ define the function $\chi_R:\FF_q\to \{0,1\}$ where for $\beta\in\FF_q$ we have $\chi_R(\beta)=1$ if $\beta\in R$ and $\chi_R(\beta)=0$ otherwise. Then $\{\chi_R\ |\ R\subseteq \FF_q,|R|=w\}\subseteq \{0,1\}^{\FF_q}$ is a $(q,(w,r))$-CFF of size ${q\choose w}$. Therefore
$$C=\{(\chi_R(b_1),\chi_R(b_2),\ldots,\chi_R(b_n))\ |\ b\in S, R\subseteq \FF_q, |R|=w\}$$ is $(w,r)$-CFF of size
$$|C|\le {q\choose w}|S|.$$

Now for the results in the table consider a constant $c>c'>1$
and let $q$ be a power of prime such that $q=c'wr+o(wr)$.
This is possible by \cite{BHP01}. By Lemma~\ref{hB} there is a hitting set $S$
for ${\cal H}_{wr}$ of size $O(wr\log n)$. This gives a $(w,r)$-CFF of size
$$O\left({q\choose w}\cdot wr\log n\right)= O\left(\left(\frac{qe}{w}\right)^w wr\log n\right)
=O\left({(ce)^wr^{w+1}} {\log n}\right)$$
that can be constructed in linear time in its size.
\end{proof}

\subsection{Separating Hash Family}\labell{SHF}\labell{Apple4}

Here we prove the following
\begin{theorem} \labell{SHFres} Let $q'>q>D_2$.
Let $S\subset \FF_{q'}^n$ be a hitting set for ${\cal H}_{D_2}$. Then
$$M(n,q,\{d_1,d_2,\ldots,d_r\})\le M(q',q,\{d_1,d_2,\ldots,d_r\})\cdot |S|.$$

In particular, for any constant $c>1$ and $q>D_2$, the following $(M; n,q$ $,\{d_1,d_2,\ldots,d_r\})$  separating hash family can be constructed in linear time
{
\begin{center}
\begin{tabular}{|c|l|c|c|c|}
\hline
 &  & poly time. & Upper & Lower\\
$n$ & $q$ & Size $=O(\ )$ & Bound & Bound \\
\hline\hline
all & $ q\ge (D_2+1)^{c}$ & $D_2\frac{\log n}{\log q}$ &$D_1\frac{\log n}{\log q}$& $D_1\frac{\log n}{\log q}$\\
\hline
all & $ q\ge {c} (D_2+1)$ & $D_2{\log n}$ &$D_1{\log n}$& $D_1\frac{\log n}{\log q}$\\
\hline
 all & $ q\ge D_2+2$ & $D_2^3{\log n}$ &$D_1\log n$& $D_1\frac{\log n}{\log q}$\\
\hline
\end{tabular}
\end{center}}
and an $(M; n,r,\{d_1,d_2,\ldots,d_r\})$  separating hash family of size
$${{cD_2\choose d_1\ d_2\ \cdots\ d_r} D_2}\log n,$$
can be constructed in time linear in the construction size.
\end{theorem}
\begin{proof} Consider the set of functions
$$\cF=\{\Delta_{(C_1,\ldots,C_r)}(x_{1},\ldots,x_{n})\ |\
C_1,\ldots,C_d \mbox{\ are pairwise disjoint},|C_i|=d_i\}$$ in $\FF_q[x_1,x_2,\ldots,x_n]$ where
$$\Delta_{(C_1,\ldots,C_r)}=
\prod_{1\le k<j\le r}\ \prod_{i_1\in C_k,i_2\in C_j}(x_{i_1}-x_{i_2}).$$
The proof then proceeds as the proof of Theorem~\ref{CFFr} and~\ref{ThH1}.
\end{proof}

\section{Open Problems}
Here we give some open problems
\begin{enumerate}
\item Find a polynomial time almost optimal
(within $poly(d)$) construction of $(n,q,d)$-PHF
for $q=o(d^2)$. Using the techniques in \cite{NSS95} it is easy to
give an almost optimal construction for $(n,q,d)$-PHF
when $q=d^2/c$ for any constant $c>1$. Unfortunately the size of the construction
is within a factor of $d^{O(c)}$ of the lower bound.

\item In this paper we gave
a construction of $(n,(w,r))$-CFF
of size
\begin{eqnarray}\label{Bsh}&&\min((2e)^w r^{w+1},(2e)^r w^{r+1})\log n\nonumber \\
&=&{w+r\choose r} 2^{\min(w\log w,r\log r)(1+o(1))}\log n\end{eqnarray}  that can be
constructed in linear time. Fomin et. al. in \cite{FLS14} gave
a construction of size
\begin{eqnarray}\label{Fom}{w+r\choose r} 2^{O\left(\frac{r+w}{\log\log (r+w)}\right)}\log n\end{eqnarray}
that can be constructed in linear time.
The former bound, (\ref{Bsh}), is better than the latter when $w\ge r\log r\log\log r$ or $r\ge w\log w\log\log w$.
We also note that the former bound, (\ref{Bsh}),
is almost optimal, i.e., $${w+r\choose r}^{1+o(1)}\log n=N^{1+o(1)}\log n,$$ where
$N\log n$ is the optimal size, when
$r=w^{\omega(1)}$ or $r=w^{o(1)}$ and the latter bound, (\ref{Fom}),
is almost optimal when
$$o(w\log\log w\log\log\log w)=r=\omega\left(\frac{w}{\log\log w\log\log\log w}\right).$$
Find a polynomial time almost optimal (within $N^{o(1)}$)
construction for $(w,r)$-CFF when $w=\omega(1)$.

\item A
construction is global explicit if it runs in deterministic polynomial time in the size
of the construction.
A local explicit construction is
a construction where one can find any bit in the construction
in time poly-log in the size of the construction.
The constructions in this paper are linear time global explicit constructions.
It is interesting to find local explicit constructions that are almost optimal.
\end{enumerate}

\end{document}